\documentclass[12pt,letterpaper]{article}
\usepackage[margin=1in]{geometry}
\usepackage{algpseudocode,algorithm,graphicx}
\usepackage{amsthm}
\pdfoutput=1 
\newtheorem {theorem}{Theorem}[section]
\newtheorem {lemma}[theorem]{Lemma}
\newtheorem{coro}[theorem]{Corollary}

\newtheorem {claim}[theorem]{Claim}

\newcommand{\pr}{\mathcal{P}}

\newcommand{\st}{\mathcal{ST}}
\newcommand{\set}[1]{\{#1\}}

\newcommand{\cp}{Closest Pair}

\newcommand{\gab}{$G_{A,B}$}
\newcommand{\diamapprox}{-Approximate Diameter}

\newcommand{\Radius}{$R$}
\newcommand{\ep}{$\epsilon$}

\title{Communication Costs in a Geometric Communication Network}
\author{Sima Hajiaghaei Shanjani \thanks{Department of Computer science, University of Victoria, Victoria, BC, Canada; sima@uvic.ca, val@uvic.ca; This research was supported by NSERC.} \and Valerie King \footnotemark[1]}
\date{March 2021}
\begin{document}

\maketitle
\begin{abstract}
 A communication network is a graph in which each node has only local information about the graph and nodes communicate by passing messages along its edges. Here, we consider the {\it geometric communication network} where the nodes also occupy points in space and the distance between points is the Euclidean distance.   Our goal is to understand the communication cost needed to solve several fundamental geometry problems, including Convex Hull, Diameter, Closest Pair, and approximations of these problems, in the asynchronous CONGEST KT1 model.   This extends the 2011 result of Rajsbaum and Urrutia for finding a convex hull of a planar geometric communication network to networks of arbitrary topology.
\end{abstract}
\section{Introduction}
The communications network is a formal way to model communication in distributed systems with an arbitrary topology. A communication network is a graph in which each node has only local information about the graph and the nodes communicate by passing messages along its edges. Here, we consider the {\it geometric communication network} where each node of the communication network occupies a point on the plane and the distance between points is the Euclidean distance.  The goal is to study the communication complexity of fundamental computational geometry problems in this setting.

Our paper extends the work of Rajsbaum and Urrutia \cite{SJ11}, the only paper known to the authors which addresses the communication complexity of a geometric communication network. That paper considers the problems of finding a convex hull and external face in an asynchronous planar network . 

Here, we consider asynchronous geometric communication networks of arbitrary topology. We show that  Convex Hull , \cp~  and Diameter problems require $\Omega(n^2)$ bits of communication. In \cp~ and Diameter our result holds even if the network is planar. Our algorithms uses $o(m)$ words of communication, where $m$ is the number of edges in the network, to approximate each of these problems.

Variants of communication networks have been studied which are distinguished by the types of messages, the existence of a global clock (synchronous vs. asynchronous), and the amount of local knowledge known to the nodes. See \cite{AGPV90,pelegtext2000}.  Unless otherwise specified, we assume the CONGEST KT1  asynchronous model which allows each node to send (possibly different) messages of size $O(\log n)$ bits to all of its neighbors at the same time.

Each node knows its position in the plane and its unique ID. To achieve sublinear in $m$ communication, we assume each node knows the name of the IDs of its neighbors  in the graph (KT1). Thus we may model a set of nodes, say in a mobile network, whose positions are changing but the topology of the underlying communication network is unchanging.  Alternatively, one might use the stronger assumption, which we call {\it fixed-position}, that the ID is equal to the name of the point on the plane, and each node knows its neighbor's position.
The model used in \cite{SJ11} does not require that each node knows the IDs of its neighbors (KT0), since there are only $O(n)$ nodes in a planar graph and each node can afford to learn its neighbor's IDs by communicating with them. To achieve its lower bounds, it requires the additional assumption that each message may contain information about only a constant number of node IDs ( "atomic" model).
 We drop this assumption in order to achieve a sublinear communication cost in dense networks.  Our lower bounds are based on the well-known two-party communication complexity of Set Disjointness \cite{KS92, KNtext97}.

 We consider the following  problems defined on a geometric communication network.  In  Convex Hull, we require that each node learns whether it is on the convex hull or not, and if so, the IDs of both its neighbours in the convex hull. For the other problems, we require that every node learns the solution. 

  \begin{itemize}
 \item 
 {\bf Convex Hull: } Compute the convex hull of the geometric communication network, which is the smallest convex shape that contains all the points. 
  
 \item 
 {\bf Diameter:}
Determine the diameter of the geometric communication network, which is the  furthest distance between any pair of points.

\item
{\bf \cp:}  Determine a pair of closest points.

\end{itemize}

In addition, we formulate the following simple problem as a building block for our lower bounds, for any function $f(x,y)$:

\begin{itemize}
\item
{\bf Path Computation of $f(x,y)$}\\
 Given a communication network 
which is the path $P=(A=v_0,v_1, v_2, ..., v_m=B)$, assume initially, A knows $x$ and B knows $y$.
After exchanging messages, all nodes know $f(x,y)$. 
\end{itemize}
\subsection{Models}

The {\it communication network} is a connected graph $G=(V,E)$, where each node in $V$ is located at a distinct point specified by $(p_x, p_y)$ where $p_x,p_y$ are integers between 1 and $n^c$, and $c$ is a constant. Each node has a unique ID of $O(\log n)$ length. In the {\it fixed point} model, the ID is given by its location;  otherwise, the ID is given by a name which may be independent of its position. Each node knows its own ID and position.
 
We use standard terminology from \cite{pelegtext2000} to describe message size and local knowledge: In the CONGEST model, each node may send a message of size $b=O(\log n)$ to every neighbor in the same time step. In the CONGEST KT1 model, each node starts with knowledge of its neighbors' IDs. In the CONGEST KT0 model, each node has a dedicated port to each of its neighbors, but is unaware of its neighbors' IDs.   We write CONGEST-$1$ when we assume the message size is 1. 

Communication may assumed to be  {\it synchronous} in that there is a global clock and messages are sent out at the start of a round and received by the end of the round, or it is {\it asynchronous}, in which case there is no global clock. Delays between sending and receiving any message may be adversarily set, though all messages sent eventually arrive. Aside from an initial wake-up, actions are event-driven. We assume all nodes wake up at the start of the algorithm. 

We assume the asynchronous model, unless otherwise stated. Time in the asynchronous model by the length of the longest chain of events (receipt of a message) in which each event is waiting for the previous event to occur. 

Nodes in the geometric network have private randomness, that is,  coins which they may flip to decide on their actions whose outcomes are not known to other nodes. The input to the algorithm is set before any random bits are known and the delay in transmitting a message may depend on the contents of all messages sent so far, including that message, but are independent of the outcomes of coin flips which have not yet been executed. 

The communication complexity of an algorithm (also referred to as a ``protocol") is measured by the expected number of bits or messages used over the worst case input and delays. The complexity of a problem is the minimum complexity over all algorithms which correctly solve the problem.

\subsection{Techniques }
There is a simple algorithm to solve any problem on the network where each node needs to learn only $O(\log n)$ bits of information. Once a spanning tree is constructed, it suffices to send the location of all points to a central location which then computes the answer and sends each node the information required. Hence any problem where each node needs to learn an $O(\log n)$ bit solution can be solved in $O(n^2)$ messages in the weakest model, asynchronous CONGEST KT0, as a spanning tree can be constructed in $O(m)$ communication and time $O(n)$.   We prove these bounds on communication are optimal within a $\log n$ factor.

We find approximation algorithms which use $o(m)$ communication, in the asynchronous CONGEST KT1 model, which enables us to initially  randomly construct  a spanning tree in $o(m)$ communication for dense networks, see Section \ref{s:stconstruction}.
The other parts of the algorithms are deterministic and can be done in the KT0 model.

The  algorithms for Approximate  Diameter and Convex Hull first finds a small set of representative points which is an \ep -kernel of the points, and use this to solve the problem. This method has been used in approximation algorithms and in streaming models  for these problems and others.   The approximate Closest Pair algorithm uses a binary search over the possible locations to locate two points close together. This type of technique is not generally useful in a streaming model with only one pass.  
 
Our lower bounds hold for the randomized asynchronous KT1 model (and therefore in the KT0 model.) The Set Disjointness problem is a well-studied problem in the literature on 2-party communication. We reduce it to our problems to prove nearly matching lower bounds for these approximation results, and as well as the lower bounds for the exact problems. Essentially, we show that each of our problems embeds a problem we call Path Communication. A protocol for this in the asynchronous CONGEST KT1 model can be used to solve Set Disjointness in the randomized two-party communication model. 

\subsection{Related Work}
For a long time, it was believed that $\Omega(m)$ bits of communication were required to build a spanning tree in a distributed message passing model. This was first proved in the CONGEST KT1 model where each message could contain information about only a constant number of node IDs ( "atomic" model)  or the node ID space is exceedingly large, beyond polynomial in size, in 1990 \cite{ AGPV90}. It was also shown in the CONGEST KT0 model, when each node does not know the ID's of its neighbors \cite{KuttenPPRT13, KuttenPPRT15}. In 2015, the first algorithm to find a minimum spanning tree  (and spanning tree) with $\tilde{O}(n)$ messages in the CONGEST KT1 model was designed for the synchronous model \cite{KKT15}, followed by an $\tilde{O}(n^{3/2})$ message-algorithm with the optimal time steps $\Theta ( D + \sqrt{n})$ in 2018  \cite{DISC18GK, DISCGP18} and an $\tilde{O}(n^{3/2})$ message algorithm for the asynchronous CONGEST KT1 model was designed in 2018 \cite{MK18,MK19}.  \cite{DISCGP18} also explores other tradeoffs between time and messages for this problem in the synchronous model. 

A natural question to ask is what other interesting graph problems can be solved in $o(m)$ communication and polynomial time in a message-passing network. Recently, Robinson proved time-communication tradeoffs in the synchronous CONGEST KT1 model for the problem of constructing a graph spanner \cite{robinson20}.  

A synchronous model for two-party communication with a global clock was introduced by Impagliazzo and Williams \cite{Impagliazzo10} time-communication tradeoffs are given. A reduction of two-party Set Disjointness has been used to prove lower bounds on time in the CONGEST model in CONGEST KT1, for example, on verification problems on graphs,  by Das Sarma et. al. in  2011 \cite{Das11, Das12}.

\subsection{Organization of the paper}

We give preliminaries including the definition of Path Computation in Section 2, and we discuss Diameter, Convex Hull, and the $\epsilon$-kernel problem in Section 3, and the Closest Pair problem in Section 4.

\section{Preliminaries: Two-party communication, paths, and trees}
In this section, we review two-party communication lower bounds, formulate and analyze the path computation problem, and review known spanning tree construction results.
\subsection{ Two-party communication}
In the basic two-party communication model (see \cite{KNtext97}), there are two players, Alice and Bob, Alice knows $x$, Bob knows $y$ and they wish to evaluate a function $f(x,y)$. The players communicate to each other via a two-way channel. The algorithm is described by a {\it protocol tree} where each internal node is labelled by a player who sends a message, each leaf is labeled with outcome for $f(x,y)$,  and the tree branches depending on the value of the string. The cost of the protocol is the number of bits sent in the worst case path from the root to the leaf. A protocol with private randomness also  contains nodes where the player flips a coin. With public randomness, there is a distribution of deterministic tree protocols. Alice and Bob choose a random string independent of $(x,y)$ which selects one of these.  If we allow randomness (either private or public) then for any given input $(x,y)$, there is a probability of reaching a particular leaf, and the expected cost of the protocol on an input $(x,y)$ is the expected number of bits needed to reach a leaf.  We may then consider the worst case expected cost of the randomized algorithm to be the worst case cost over all inputs. The communication complexity of a problem is the maximum over all inputs of the expected cost.

The basic two-party model assumes messages arrive as soon as they are sent but there is no notion of a global clock, so that information is not gained if Alice waits or sends an empty message. In a 2010 paper, Impagliazzo and Williams introduced the {\it synchronous bit two-party communication model} in which there is a global clock. The levels of the protocol tree represent time steps, and at each time step, a party may send a 0,1, or $*$ , where $*$ means no message is sent. Cost is measured by the number of $0$'s and $1$'s sent. It is easy to see that Alice can communicate the contents of an $n$ bit string by sending a 1 at time equal to the value of the string and no other messages. There are clearly communication, time tradeoffs.  For example, Alice can send an $n$ bit string by sending one bit in exponential time equal to the value of the string. Or, in general with time polynomial in $n$, Alice can send an $n$ bit string using  a cost of $n/\log n$ in this model. See \cite{Impagliazzo10, robinson20}. Hence it does not make sense to talk about a lower bound on communication costs here without an upper bound on the time.
For the remainder of this section, we prove lower bounds on communication cost in models with no global clock. 

The following well-studied problem will be used as a basis for proving lower bounds in this paper:

\begin{itemize}
\item {\bf Set Disjointness:}
Alice has set $A$ and Bob has set $B$, $A$ and $ B \subset \{1,2,...,n\}$. Give a two party communication protocol to answer if $A\cap B=\emptyset$.
\end{itemize}
No algorithm can do better asymptotically than the one where Alice sends the $n$ bit characteristic vector of $A$ to Bob, and Bob returns the one bit solution to Alice.
\begin{theorem}{\cite{KS92,RAZBOROV}}\label{t:setdisjoint}
Alice has set $A$ and Bob has set $B$, $A$ and $ B \subset \{1,2,...,n\}$. Any two party communication protocol  to answer if $A\cap B=\emptyset$ has cost at least $\Omega(n)$, even if there is public randomness.  \end{theorem}

\subsection{ Path Computation }
For any two-party function $f$, we can define a problem in a communication network consisting of a path.
\begin{itemize}
\item
{\bf Path Computation of $f$}\\
We are given a function $f(x,y)$ and a communication network 
which consists of a path $P=(v_0,v_1, v_2, ..., v_m)$, 
where $v_0= $ Alice,  $v_m= $ Bob. Alice knows $x$ and Bob knows $y$.
Assume Alice and Bob awake at the start. Nodes pass messages until  every node in $P$ learns $f(x,y).$
\end{itemize}


\begin{theorem} \label{t:2-to-path}
Let $C_2(f)$  be the randomized complexity of $f$ in the basic two-party communication model. Then 
 the randomized complexity of the Path Computation of $f$ in the CONGEST-$b$ KT1 model where each message has $b$ bits
  is at least $\Omega(\frac{m}{b} (C_2(f)- \lg m))$.

\end{theorem}

We first prove this for $b=1$ and for a model with a restricted adversary which we call the restricted model or {\it r-asynchronous} CONGEST-1 KT1. The restricted model takes away  adversary control over the order in which messages are received, while there is still no global clock, as in the basic two-party communication model. Hence the adversary only controls the inputs. E.g,  
 messages are received in the order they are sent, where there is choice as to which of two messages is received from two different neighbors, the message from the lower numbered node is received first, and if these rules leave a choice as to which of two messages to be sent first in opposite directions over the same link, the message from the lower numbered node is sent first.  
The nodes in the path have ID's $(0,1,...,m)$.
In the proofs below, the term ``path protocol" refers to an algorithm to solve Path Computation of $f$, in the r-asynchronous CONGEST-$b$ KT1 model. Now we show:

\begin{lemma} \label{l:one-input}
 For every path protocol in the r-asynchronous CONGEST-$1$ KT1 model, for every edge, there is at least one input which requires an expected $C_2(f)$  bits to be communicated over that edge. \end{lemma}
\begin{proof}
Assume by contradiction that $e=\{i,{i+1}\}$ is the edge closest to 0 which carries an expected $l< C_2(f)$ bits of communication for all input.  Then we claim there is a protocol $\pr$ on the basic 2-party model which uses less than $C_2(f)$ bits of communication for all input.

We prove this by induction on the time $t$:

Initially, Alice and Bob both know $i$ and they know the states of $0,1,...,i$ and ${i+1},{i+2},..,m$ respectively.
Assume by induction that Alice knows the states of $0,1,...,i$ and Bob knows the states of $i+1,i+2,...,m$ nodes up to time $t$, where time is measured by the total number of messages sent so far. 

Let $[j,k]$ denote the path between nodes $j$ and $k$.
Suppose at time $t$, $\pr$ sends a message across an edge from the ``sender" to a ``receiver".  If the sender and the receiver are in the path $[0,i]$, then since Alice knows the states of both the sender and the receiver, Alice can simulate this step in the protocol by determining the message the sender would send (flipping a private coin if the sender would do so) and updating the state of the receiver after that message is received.  Bob does this similarly if the sender and the receiver are in $[i+1,m]$.  
If the sender is $i$ and the receiver is $i+1$, Alice sends the message to Bob, and Bob updates the state of $i+1$ in its simulation, and vice versa, if the sender is $i+1$ and the receiver is $i$.  

No more than an expected $l$ bits are exchanged between Alice and Bob.  Each knows the state of at least one node in the path and therefore knows $f$. 
The exchange of messages between Alice and Bob is a two-party protocol: a message is sent and received immediately after, one at a time, as in a protocol tree; Alice's messages depend only on its state which is determined by $x$ and the messages previously received from Bob and similarly for Bob.  Hence, this contradicts the $C_2(f)$ lower bound.  \end{proof}

 \begin{lemma}\label{l:random}
 Suppose there is a path protocol $\pr$  such that in the r-asynchronous CONGEST-$1$ KT1 model for $f$ on input $I$ the expected communication cost  is no greater than $C_{\pr}(I)$ for a path of length $m$. Then there is a randomized path protocol $\pr'$ such that the expected number of bits communicated across any edge is no greater than $ C_{\pr}(I)/m + \lceil{ \lg m}\rceil$. 
 \end{lemma}

 \begin{proof}
We will define a protocol $\pr'$ which simulates $\pr$. We assume the first and last nodes wake up. 
\begin{enumerate}
\item 
Node 0  randomly picks a number $a \in \{0,1,...,m-1\}$ and communicates $a$ down the path to node $m$ using $\lceil{\lg m}\rceil$ bits per edge.
\item Node 0  simulates the
 nodes in $[0,a]$ and forwards whatever message $M$ that $a$ would have sent to $a+1$ down the path to $m$.
\item  After node $m$ receives the value $a$, it simulates the nodes in  $[a+1, m]$. When $a+1$ would send a message $M'$ to $a$ in the simulation, node $m$ sends this message down the path to node $0$.  
\item All nodes $a \in \{1,...,m-1\}$ upon receiving a message from $a-1$ passes it to node $a+1$ and upon receiving a message from $a+1$ passes it to $a-1$. 
\end{enumerate}

Let $C_{\pr}(I)$ denote the expected communication cost for input $I$ under protocol $\pr$. We can write $C_{\pr}(I)= \sum_{i =0}^{m-1} c_{\pr}(i,+,I ) + c_{\pr}(i+1,-,I )$, 
where $c_{\pr}(i, +,I)$ denotes the expected number of bits sent from node $i$ to node $i+1$ under protocol $\pr$ on input $I$, and, similarly, $c_{\pr}(i+1,-,I)$ denotes the expected number of bits received from node $i $ from node $i-1$.   

We observe that when $\pr'$ passes $b$ bits from node 0 to node $m$ or from node $m$ to node 0, the bits arrive at $m$ and the communication cost over all the edges is $mb$.
Then, by this observation, for any input $I$,  for a given choice of $a$, and for any edge in the path, under $\pr'$, the communication per edge (in both directions) is no greater than $\lceil{\lg m}\rceil$ plus the communication cost  $c_{\pr}(a, +,I) + c_{\pr}(a+1,-,I)$. Taking the expectation over all choices of $a$ chosen uniformly at random, we have
that the expected number of bits sent over any one edge  $\{i,i+1\}$ in both directions by $\pr'$ on input $I$, $c_{\pr'}(i, +,I)+c_{\pr'}(i,-,I)$ is: 
\begin{eqnarray*}
& \leq & \lceil{\lg m}\rceil + \sum_{i=0}^{m-1} Pr (a=i) [c_{\pr}(i,+,  I) +  c_{\pr}(i+1,-,  I) ] \\
 & =&  \lceil{\lg m}\rceil + 1/m [\sum_{i=0}^{m-1}  c_{\pr}(i,+,  I) +  c_{\pr}(i+1,-,  I)]\\
   &  = & \lceil{\lg m}\rceil +   (1/m) C_{\pr}(I) 
 \end{eqnarray*}
 This completes the proof of the lemma. \qed

The proof of Theorem \ref{t:2-to-path} follows:
By Lemma \ref{l:one-input}, we know that for any protocol $\pr'$ and any edge $\{i,i+1\}$, there is some $I$ such that the expected number of bits passed in both directions over that edge  $c_{\pr'}(i,+, I)+c_{\pr'}(i+1,-,I) \geq C_2(f)$.
Hence by Lemma \ref{l:random}, $(1/m)C_{\pr}(I) +  \lceil{\lg m}\rceil  \geq C_2(f)$ which implies that $C_{\pr}(I) \geq m (C_2(f)  - \lceil{ \lg m}\rceil)$. Let $\pr$ be the optimal algorithm for the path computation of $f$ and $I$ be its worst case input, then the communication complexity for the r-asynchronous CONGEST KT1 model $C_{m,r}(f) \geq C_{\pr}(f,I) \geq m (C_2(f)- \lceil{ \lg m }\rceil)$.  

We observe that the lower bound for the asynchronous CONGEST-1 KT1 model $C_{m}(f) \geq  C_{m,r}(f)$ since the restricted model can be simulated by an adversary strategy in the non-restricted model. 
If we consider the CONGEST-$b$ KT1 model, then as $b$ bits are communicated with each message, we can conclude that the complexity of computing $f$ in this model is at least $C_m(f)/b$ or $\frac{m}{b} (C_2(f)- \lceil{\lg m}\rceil)$.\end{proof} 

\subsection{Spanning Tree Construction}\label{s:stconstruction}
Here, we review what is known concerning the construction of spanning trees in a communications network as these will be the building blocks for our algorithms. Note that both algorithms are randomized and succeed with high probability.
\begin{itemize}
\item
In CONGEST KT1, there is a synchronous algorithm which runs in time and communication $\tilde{O}(n)$ and an asynchronous model with communication $\tilde{O}(\min\{n^{3/2}, m\})$. 
\item
In CONGEST KT0, $\Omega(m)$ messages are required to construct a spanning tree in the synchronous (and asynchronous) model.
\end{itemize}

With the exception of the subroutine to build the spanning tree, our algorithms are  designed to work in the weakest model described in this paper, the asynchronous CONGEST KT0 model. We state the communication cost of each of our algorithms as a function of the communication needed to build the spanning tree on $n$ nodes, plus the communication needed for the other part of the algorithm.

\section{Diameter and Convex Hull}
We  present the $\epsilon$-kernel problem, which we use as a subroutine to compute approximations for Diameter and Convex Hull. Then we show lower bounds of $\Omega(n^2)$ bits on the expected communication needed for exact solutions for Diameter and Convex Hull. 

\subsection{\ep -Kernel}

Let $\mathcal{S}^{d-1}$ denote the unit sphere centered at the origin in $\mathcal{R}^d$. For any set $P$ of points in $\mathcal{R}^d$, a point $u \in \mathcal{S}^{d-1}$ is called a {\it direction}. Given a set of points $P\in R^d$, the {\it directional width of $P$ in direction $u$}, denoted $w(u,P)=max_{p\in P} <u,p> - min_{p\in P} <u,p> $, where $<.,.>$ is the standard inner product. An  \ep -kernel of $P$ is a subset of $P$ which approximates $P$ with respect to directional width. More formally: A subset $Q \subset P$ is called an \ep -kernel of $P$ if for each $u \in \mathcal{S}^{d-1}$, 
$(1-\epsilon)w(u,P) \leq w(u,Q)$ \cite{AgarwalHV04, Agarwal05}.   In this paper, $d=2$.

Algorithm \ref{a:eps-kernel} computes an \ep -kernel of the set of points in  $\mathcal{R}^2$ by using the idea of rounding directions, which was introduced in \cite{Agarwal92}.

Let $loc(p)$ denote the position of node $p$. Let $\epsilon < 1$ and $\delta= \sqrt{2\epsilon}$. In Algorithm \ref{a:eps-kernel}, $A$ is the set of lines passing through $(0,0)$ which, with the $x$-axis, make angles $\delta * i$ for  $i=1, ..., \lceil{\pi/\delta}\rceil$, and $|A|=\theta(1/\sqrt{\epsilon})$. Given a set of points $P$ and a line $a$ through the origin, the {\it extreme points for the line } is the pair of points $r$ and $s$ in $P$ whose projections onto the line are furthest apart on the line: $r$ is a point such that $<r,u>=max_{p\in P} <u,p>$ and $s$ is a point such that $ <s,u>=min_{p\in P} <u,p> $, where $<.,.>$ is the standard inner product.
 
A spanning tree $ \st$ is first constructed. Starting with the leaves of $\st$, each node compares its own location with the location of the extreme points received from its children (if it is not a leaf) to find the best two candidates for the extreme points for each direction, and sends them to its parent, until eventually the root receives all the extreme points. The set of these extreme points is a \ep -kernel for this set of points.

 \begin{algorithm}
 
\caption{Asynchronous $\epsilon$-Kernel\label{a:eps-kernel}}
\begin{algorithmic}[1]
\Procedure{$\epsilon$-Kernel}{P is a set of $n$ nodes} 
  \State Find a spanning tree $\st$ with leader $L$.
  \State $L$ broadcasts  $<start>$ when $\st$ is complete. 
  \State When a leaf $p$ receives $<start>$,  it sends $S_p=\{loc(p)\}$ to its parent.
  \State When a non-leaf node $p\neq L$ receives $S_q$ from all of its children $q$, $p$ computes the set $S_p$ of extreme points for all lines in $A$ among the point set $\bigcup_{q} S_q \cup \{loc(p)\}$ and passes $S_p$ to its parent.
    
  \State When $p= L$ receives $S_q$ from all of its children $q$, $L$ computes the set $S_L$ of extreme points for all lines in $A$ among the point set $\bigcup_{q} S_q \cup \{loc(L)\}$.
 
      \EndProcedure
\end{algorithmic}
\end{algorithm}

\begin{theorem}

$S_L$ is an \ep -kernel for the point set $P$.
\end{theorem}
\begin{proof}
Given any pair of points $p$, $q$, let $d(p,q)$ denote the distance between points $p$ and $q$, and let $p_a$ denote the projection of point $p$ onto line $a$. Let $\mathcal{S}$ denote the unit sphere centered at the origin in $\mathcal{R}^2$, we need to show that for any direction $a^* \in\mathcal{S}$, there is a direction $a \in A$ such that for any $p,q \in P$, $$(1-\epsilon)d(p_{a*},q_{a*}) \leq d(p_a, q_a)$$ 

So it is sufficient to show $(1-\epsilon)d(p,q) \leq d(p_a, q_a)$, as $d(p_{a*},q_{a*}) \leq d(p, q)$.

Let $\overrightarrow{pq}$ denote the line passing through $p$ and $q$, and let $\theta(p,q,a)$ be the acute angle formed by $\overrightarrow{pq}$ and $a$.
Then it is not hard to observe (see figure \ref{Fig:epkernel}):

$$ d(p_a,q_a) \leq d(p,q) \leq d(p_a, q_a) /cos(\theta(p,q,a))  $$

As there is a line forming an angle with the $x$-coordinate at every $\delta$ angle interval, there must be some line $a \in A$ forming an angle $\leq \delta$ with $\overrightarrow{pq}$. 

Finally, we observe that the Taylor series implies
$cos(\delta)\geq 1 -\delta^2/2$ . In  particular, for $\delta=\sqrt{2\epsilon}$, then $cos(\delta) \geq 1-\epsilon$.\end{proof}

\begin{figure}

\centering
\includegraphics[scale=0.23]{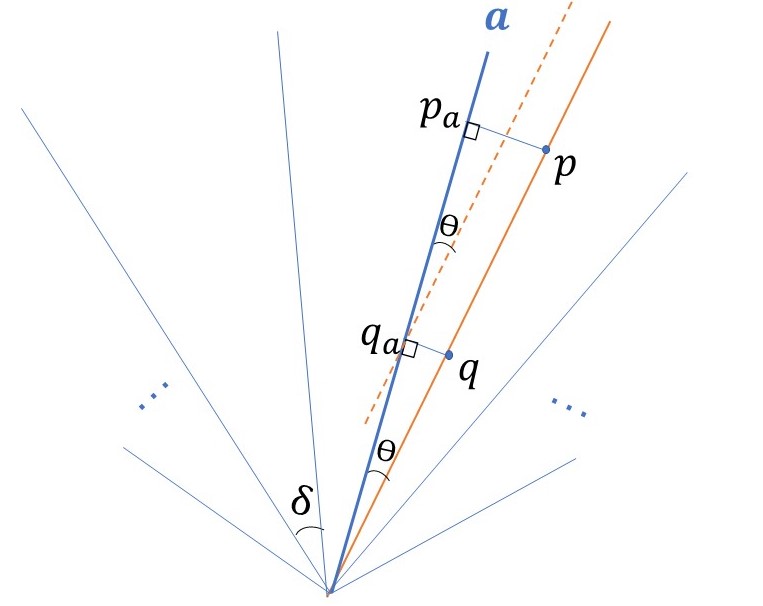}
\caption{\ep - kernel proof \label{Fig:epkernel}}
\end{figure}

\noindent
{\it Analysis of communication and time complexity of Algorithm \ref{a:eps-kernel}:}  Each node passes up no more than $2|A|=O(1/\sqrt{\epsilon})$ locations using  $O(diam/\sqrt{\epsilon}
)$ time and $ O(n/\sqrt{\epsilon})$ words  of size $O(\lg n)$ where $diam(\st)$ is the diameter of  $\st$.    In addition, there is the initial communication cost of computing the spanning tree. 

\begin{theorem}
\label{thdiamapprox}
There exists an asynchronous algorithm in CONGEST KT1 which computes an $\epsilon$-kernel in a graph of $n$ nodes with $O(\frac{n}{\sqrt{\epsilon}})$ messages in time $O(\frac{diam(\st)}{\sqrt{\epsilon}} )$ plus the costs of constructing $\st$.   
\end{theorem}

\subsection{Approximate Diameter}
 \begin{algorithm} 
 
\caption{Asynchronous $\epsilon$-Diameter \label{a:ep-diam}}

\begin{algorithmic}[1]
\Procedure{$\epsilon$-Diameter}{P is a set of $n$ nodes} 
  \State Find \ep -kernel $S$ and spanning tree $\st$ with leader $L$.
  \State If $p=L$, $L$ computes the two farthest points $r$ and $s$ in $S$ and broadcasts $r$ and $s$ through $\st$.
 
      \EndProcedure
\end{algorithmic}
\end{algorithm}

\begin{theorem}
\label{thdiamapprox1}
There exists an asynchronous algorithm in CONGEST KT1 which computes a $(1-\epsilon)$-\diamapprox~ in a graph of $n$ nodes with $O(\frac{n}{\sqrt{\epsilon}})$ messages in time $O(\frac{diam(\st)}{\sqrt{\epsilon}})$ plus the costs of computing $\st$. 
\end{theorem}
\begin{proof}
In Algorithm \ref{a:ep-diam}, once the \ep -kernel $S$ is received by the leader, the leader finds the two points furthest from each other in $S$ and broadcasts them. The diameter of $S$ is a $(1-\epsilon)$-approximation for the diameter of $P$ by the definition of $\epsilon$-kernel. More formally, let $p_1$ and $p_2$ be the points in $P$ furthest from each other. Assume $p_1$ and $p_2$ are extreme points of $P$ for some direction $a* \in \mathcal{S}^{d-1}$. As $S$ is an \ep -kernel of $P$, then for each $u \in \mathcal{S}^{d-1}$ including $a*$, 
$(1-\epsilon)w(u,P) \leq w(u,S)$, where $w(u,p)$ is the directional width of $P$ in direction $u$: $w(u,P)=max_{p\in P} <u,p> - min_{p\in P} <u,p> $.

The number of messages in this algorithm is the number of messages we need to compute \ep -kernel plus $O(n)$ messages to broadcast the two farthest points.\end{proof} 

\subsection{Approximate Convex Hull}
 There are different notions of {\it approximate convex hull}. An \ep -hull is one of the commonly used notions and it is closely related to \ep-kernel.

Let $P$ be a set of $n$ point in $\mathcal{R}^d$. Let $C(P)$ denotes the convex hull of $P$. An {\it \ep -hull} $S$ of $P$ such that all the points in $P$ are either in $C(S)$ or within distance {\ep} from $C(S)$. 

There have been extensive studies for \ep -kernel and \ep -hull. Much effort has been done to find a fast algorithm that computes a small \ep -kernel in the sequential models. There are also extensive work for approximating Convex Hull in the streaming model \cite{Suri, AgarwalHV04, Kumar17, Chan2006, Arya}.

An {\it \ep -kernel} is a multiplicative error version of  \ep -hull. Let $Diam(P)$ be the diameter of the point set $P$. If we have an approximation of a diameter, then we can set $\epsilon$ in a way that \ep -kernel will give us a desired $\epsilon*Diam(P)$-hull \cite{Agarwal05}.

\begin{algorithm} 
\caption{Asynchronous $\epsilon$-Approximate Convex Hull \label{a:ep-ch}}

\begin{algorithmic}[1]
\Procedure{$\epsilon$-Convex Hull}{P is a set of $n$ nodes} 
  \State Find \ep -kernel $S$ and spanning tree $\st$ with leader $L$.
  \State If $p=L$, $L$ broadcasts the vertices of convex hull of $S$ through $ST$.
 
      \EndProcedure
\end{algorithmic}
\end{algorithm}

\begin{theorem}
\label{thchapprox}
There exists an asynchronous algorithm in CONGEST KT1 which computes an $\epsilon$- Approximate Convex Hull and $f(\epsilon)-hull$ in a graph of $n$ nodes with $O(\frac{n}{\sqrt{\epsilon}})$ messages in time $O(\frac{diam(\st)}{\sqrt{\epsilon}})$ plus the costs of computing $\st$, where $f(\epsilon)=\epsilon*Diam(P)$. 
\end{theorem}
\begin{proof}
In Algorithm \ref{a:ep-ch}, once the \ep -kernel $S$ is received by the leader, the leader finds the convex hull of $S$ and broadcasts its vertices, so that every node knows whether it is on the convex hull and knows its neighbors in both the clockwise and the counterclockwise directions.
The \ep -kernel is a multiplicative error version of \ep -hull \cite{Agarwal05}. By definition, \ep -kernel $S$ approximates the convex hull of $P$ within a  $1-\epsilon$ factor in any direction. 

The convex hull of $S$ is an $\epsilon* Diam(P)$-hull, as all the points in $P$ are either in the convex hull of $S$ or within $\epsilon Diam(P)$ from the convex hull of $S$, where $Diam(P)$ is the diameter of $P$. This is because if $d(x,S)$ is the distance of any point $x\in P \setminus S$ from convex hull $S$, then $d(x, S) \leq w(u,p)-w(u,S)\leq \epsilon w(u,p)$ for any direction $u \in \mathcal{S}^{d-1}$, where $w(u,p)$ is the directional width of $P$ in direction $u$. Thus $ d(x,S)\leq \epsilon*Diam(P)$, and the \ep -kernel is an $\epsilon*Diam(P)$-hull. \end{proof}

\subsection{Lower bound for diameter}
\label{sec:lowerbounddiam}
 
We reduce the problem of Path Set Disjointness to Diameter. Similar ideas have been used to show lower bounds in a streaming model \cite{lowerbound}.  
Given a path of $m=n$ edges with endpoints Alice and Bob, where Alice and Bob each know subsets  $A$ and $B$ resp., of $\{1,...,n\}$, the communication network must compute the Set Disjointness problem on $(a,b)$, where $a=(a_1,...,a_n)$ and $b=(b_1,...,b_n)$ are the characteristic vectors of $A$ and $B$.  Consider the following   graph \gab$(V,E)$ where $|V|=3n$. 
Let $V= W\cup U\cup H$, $W=\set{w_1,w_2,...,w_n}$, $U=\set{u_1,u_2,...,u_n}$, $H=\set{h_1,h_2,...,h_{n}}$.
  
We draw evenly spaced lines $L_1,...,L_n$ through the center of the $[n^c] \times [n^c]$ grid, for a constant $c\geq 2+2/\lg n$. Each line $L$ is composed of a {\it low} ray $l(L)$ which goes from the center to below the center and the {\it upper} ray  $u(L)$ which goes from the center to above the center. The lines are drawn so that the upper rays are at equally spaced angles from each other and between 0 and $\pi$ with the $x$ axis. The points of the Diameter problem are determined as follows (Figure \ref{Fig:diameter1}), where $r=n$ and $R=2n^2$. 

\begin{itemize}
\item For each $1\leq i \leq n$, if $a_i=0$, $w_i$ is on $u(L_i)$ at distance $r$ from the center and otherwise at distance $R$. 
\item For each $1 \leq i \leq n$, if $b_i=0$, $u_i$ is on $l(L_i)$ at distance $r$ from the center and otherwise at distance $R$.
\item The $h_i$'s, where $1\leq i \leq n$, are located on the grid points inside the circle centered at the origin with radius $r$. 

\end{itemize}

In the above description of the location of the points, in the case that any coordinate is not an integer coordinate, we round the point as follows:
For each $1\leq i \leq n$, if the angle of $u(L_i)$ with the $x$ axis is less than or equal $\pi/2$, then we round $w_i$ to the top right corner of the cell that contains $w_i$, and if the angle is more than $\pi/2$ we round $w_i$ to the top left corner of the cell that contains $w_i$. For each $1\leq i \leq n$, if the angle of $l(L_i)$ with the $x$ axis is less than or equal $3\pi/2$, then we round $w_i$ to the bottom left corner of the cell that contains $u_i$, and if the angle is more than $3\pi/2$ we round $u_i$ to the bottom right corner of the cell that contains $u_i$.

Note that a circle centered at the origin with radius $r$ contains at least $\pi r^2$ lattice points \cite{phdLevin}. So choosing $r=n$ guarantees that there are more than $n$ grid points inside the circle to locate points of $H$. As $R=2n^2$, no two points of $W$ and $U$ are in the same grid cell and they would not round to the same grid point.

\begin{figure}

\centering
\includegraphics[scale=0.2]{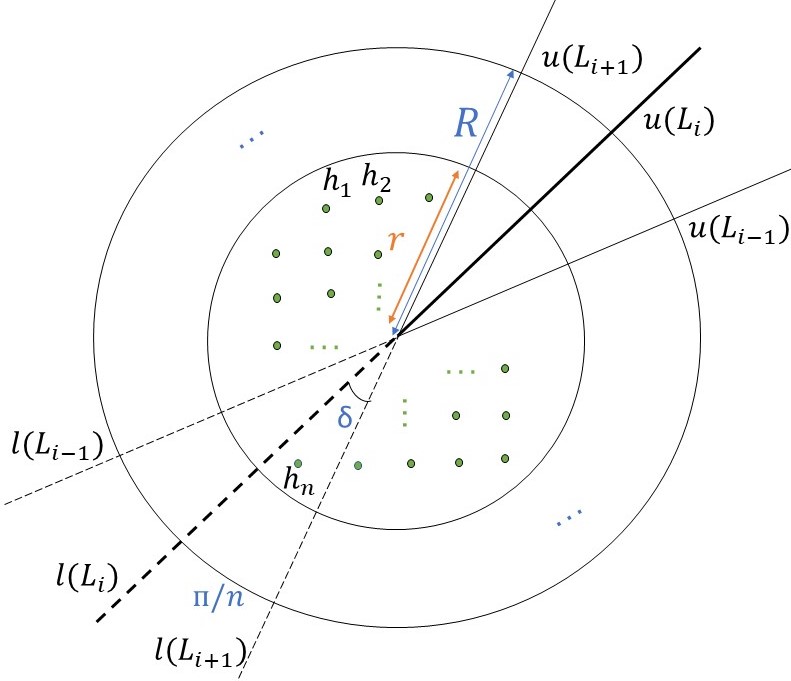}
\includegraphics[scale=0.2]{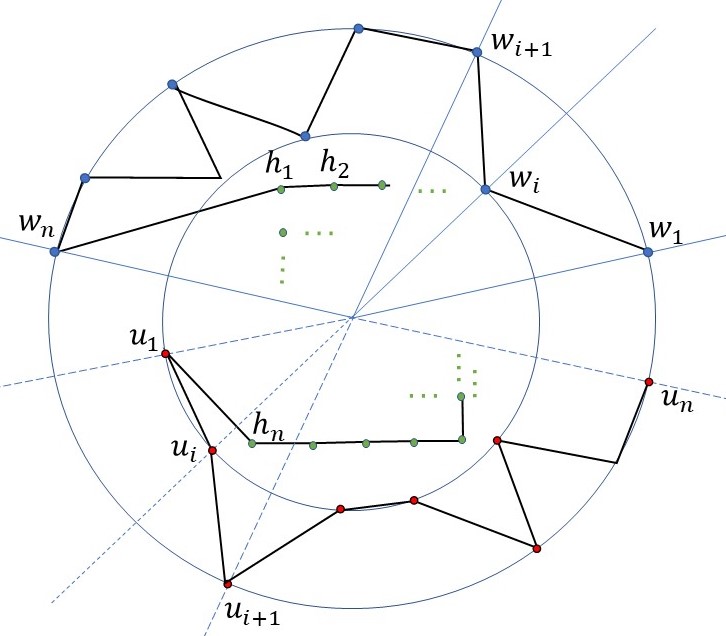}
\caption{Reduction from Set Disjointness to Diameter \label{Fig:diameter1}}
\end{figure}
We next describe the edges.
\begin{itemize}
\item
For each $1\leq i \leq n$, $E$ contains  $\{w_i, w_{i+1}\}$ and $\{u_i, u_{i+1}\}$.
\item
 For each $1\leq i \leq n$, $E$ contains $\{h_i, h_{i+1}\}$. 
 \item
$E$ contains $\{w_n, h_1\}$ and $\{h_{n}, u_1\}$.
\end{itemize}

\begin{claim}
The diameter of \gab is {2\Radius} if and only if $A\cap B\neq \emptyset $.
\end{claim}
\begin{proof} Let $C_R, C_r$ be the circles centered at the center of the $[n^c] \times [n^c]$ grid at distance $R=2n^2$ and $r=n$, respectively, for a constant $c\geq 2+2/ \lg n$.

For any pair of points in $U \cup W$, the distance is less than $2R$ if and only if one of the following is true: (1) the two points are on two different lines; (2) the two points are on the same line, one at distance $r$ and the other distance $R$. Note that even if the points are rounded, this still holds. This is because, the length of the longest chord after diameter is not more than $\cos{\alpha} +2\sqrt{2}$, where $\alpha=\pi/2n$, as shown in Figure \ref{Fig:diamapproxhard}. We assumed $R=2n^2$, so $\cos{\alpha} +2\sqrt{2}<2R$ by the Taylor series.
If there is an $L_i$ such that both $w_i$ and $u_i$ are distance $R$ from the center, then their distance is at least $2R$. This occurs only when for some $i$,  $a_i=b_i=1$, so $A\cap B\neq \emptyset $.\end{proof}

\begin{theorem} 
\label{thm:diameterlowerbound}
Any randomized asynchronous distributed algorithm in the CONGEST model for solving Diameter in a graph of $n$ nodes requires $\Omega(n^2)$ expected bits of communication.
\end{theorem}
 \begin{proof}
 Suppose  Alice and Bob are two endpoints of a path with $N=n/3$ edges numbered $0,1,...,N$, where Alice and Bob each know subsets  $A$ and $B$ resp., of $\{1,...,N\}$, and they want to compute the Set Disjointness problem on $(a,b)$, where $a=(a_1,...,a_N)$ and $b=(b_1,...,b_N)$ are the characteristic vectors of $A$ and $B$.  
  Alice can construct the subgraph of $G$ induced by $H \cup W$ and given $B$, Bob can construct the subgraph induced by $H \cup U$.  They can then each carry out a simulation of the protocol which determines the diameter  of $G$ and output 1 or 0 depending on whether the diameter is at least {2\Radius} or not.  Thus the communication cost is no less than the cost of computing disjoint sets over a path of length $N$ where $n=3N$ is the total number of nodes in $G$. By Theorem \ref{t:2-to-path},  $\Omega(n^2/\lg n)$ messages are required.\end{proof}

\begin{theorem} 
Any randomized asynchronous distributed algorithm for approximating Diameter in a graph of $n$ nodes within a $1-\epsilon$ factor of optimal requires $\Omega( min\{n^2,1/\epsilon\})$ expected bits of communications.
\end{theorem}

\begin{proof}
Consider the reduction from Set Disjointness to Diameter described before with vertices $V=W\cup U \cup H$ as shown in Figure \ref{Fig:diameter1}. Let $C_R, C_r$ be the circles centered at the center of the $[n^c] \times [n^c]$ grid at distance $R=2n^2$ and $r=n$, respectively, for a constant $c\geq 2+2/\lg n$. The maximum distance between any pair of points is the diameter of the largest circle, which is at least $2R$. The next greatest distance in this structure is $d \approx 2R \cos{\alpha}$, as this is the longest chord after diameter on the largest circle and it is greater than the diameter of the smaller circle, which is $r=n$, Figure \ref{Fig:diamapproxhard}. The reason that we say $d \approx 2R \cos{\alpha}$ is because we are rounding the points to grid points. As no more than one point is in one cell and the points are rounded to the outside of the circle,  the rounding can only increase the distance of the points from the center of the circle. This means the ratio of $d/2R$ after the rounding is not more than $\frac{2R \cos{\alpha} +2\sqrt{2}}{2R}=\cos{\alpha}+\frac{2\sqrt{2}}{2R}$. As we assumed $R=2n^2$, the ratio of $d/2R$ after the rounding is not more than $\cos{\pi/2n}+\sqrt{2}/n^2$.

Therefore, if there is a $(1-\epsilon)$-approximation algorithm for the Diameter problem, where $1- \epsilon > \cos{\pi/2n}+\sqrt{2}/n^2$, then the algorithm also finds the exact solution of Diameter and solves the problem of Set Disjointness of size $n/3$. So, if $\epsilon <1-\cos \pi/2n-\sqrt{2}/^2n$, then by Theorem \ref{thm:diameterlowerbound} any randomized asynchronous distributed algorithm for approximating Diameter of $n$ points within $1-\epsilon$ factor of optimal requires $\Omega(n^2)$ bits of communications.

If $\epsilon \geq 1-\cos(\pi/2n)-\sqrt{2}/n^2$, we build the same structure but with $\theta(1/\sqrt {\epsilon})$ points, where $R=2n^2$, and reduce a Set Disjointness problem with size $n'=1/2\sqrt{\epsilon}$ to this problem. 
By using Taylor series we know that $1-\cos(x)\geq x^2/2!-x^4/4! \geq 11x^2/24$, if we pick $n'=1/2 \sqrt{\epsilon}$, then $1-\epsilon > \cos\pi/2n'+\sqrt{2}/n^2$ as $1-\cos\pi/2n'\geq 11(\pi/2n')^2/24\geq  4\epsilon$ and $\epsilon-\sqrt{2}/n^2>0$. Thus, if $\epsilon >1-\cos \pi/2n-\sqrt{2}/n^2$, then by Theorem \ref{thm:diameterlowerbound} any randomized asynchronous distributed algorithm for approximating Diameter of $n$ points within $1-\epsilon$ factor of optimal requires $\Omega(1/{\epsilon})$ bits of communications.\end{proof}

\begin{figure}
\centering
\includegraphics[scale=0.22]{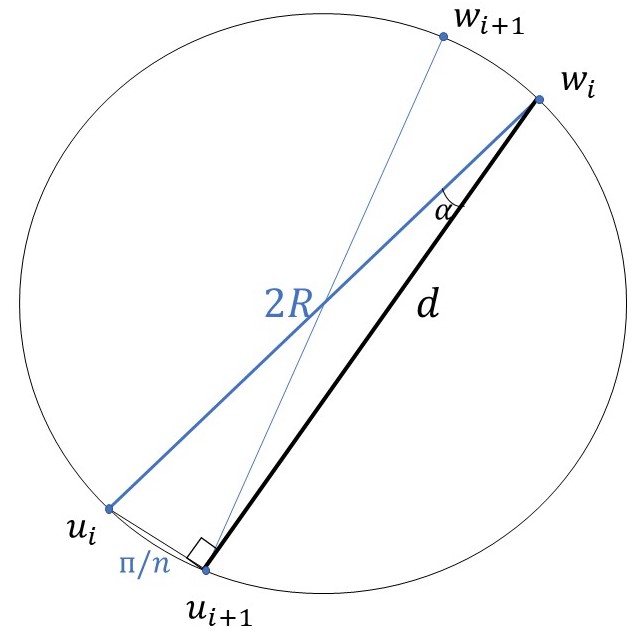}
\caption{The distance between $w_i$ and $u_{i+1}$ is the second largest distance of a pair of points after the diameter. 
\label{Fig:diamapproxhard}}
\end{figure}

\subsection{Lower Bound for Convex Hull}

We reduce the problem of Path Set Disjointness to Convex Hull. We show that if Convex Hull can be solved in $o(n^2)$ communication in an asynchronous network with $n$ nodes, then this implies an algorithm for Path Set Disjointness on a path of length $n/3$ with fewer than $\Omega(n^2)$ expected bits of communication, giving a contradiction.

Given a path of $m=n$ edges with endpoints Alice and Bob, where Alice and Bob each know subsets  $A$ and $B$ resp., of $\{1,...,n\}$, the communication network must compute the Set Disjointness problem on $(a,b)$, where $a=(a_1,...,a_n)$ and $b=(b_1,...,b_n)$ are the characteristic vectors of $A$ and $B$.  Consider the following   graph \gab$(V,E)$ where $|V|=3n$. 
Let $V= W\cup U\cup H$, $W=\set{w_1,w_2,...,w_n}$, $U=\set{u_1,u_2,...,u_n}$, $H=\set{h_1,h_2,...,h_{n}}$.

We consider $4n$ positions equally distanced on a circle centered at the center of the $[n^c] \times [n^c]$ grid, for a constant $c>1+1/(2\lg n)$, with radius $R=2n$: $P=\{p^i_j|1\leq i \leq n, 1\leq j \leq 4$, as shown in Figure \ref{Fig:convex}. In the case that any coordinate is not an integer coordinate, we round the point to one of grid points of the cell that contains this point as mentioned before in Section \ref{sec:lowerbounddiam}. The points of vertices $V$ for the the Convex Hull problem are determined as follows (Figure \ref{Fig:convex}):

\begin{itemize}
\item For each $1\leq i \leq n$, if $a_i=0$, $w_i$ is on the location $P^i_0$; otherwise on the location $P^i_1$; 
\item For each $1 \leq i \leq n$, if $b_i=0$, $u_i$ is on the location $P^i_2$; otherwise on the location $P^i_3$; 
\item The $h_i$'s, where $1\leq i \leq n$, are located on the grid points inside the circle centered at the center of the grid with radius $r=n$. 

\end{itemize}

Note that a circle centered at the origin with radius $r$ contains at least $\pi r^2$ lattice points \cite{phdLevin}. So choosing $R>n$ guarantees that there are more than $n$ grid points inside the circle to locate points of $H$. We choose $R=2n$ to make sure no two points of $W$ and $U$ are in the same grid cell and would not round to the same grid point.

\begin{figure}
\centering
\includegraphics[scale=0.24]{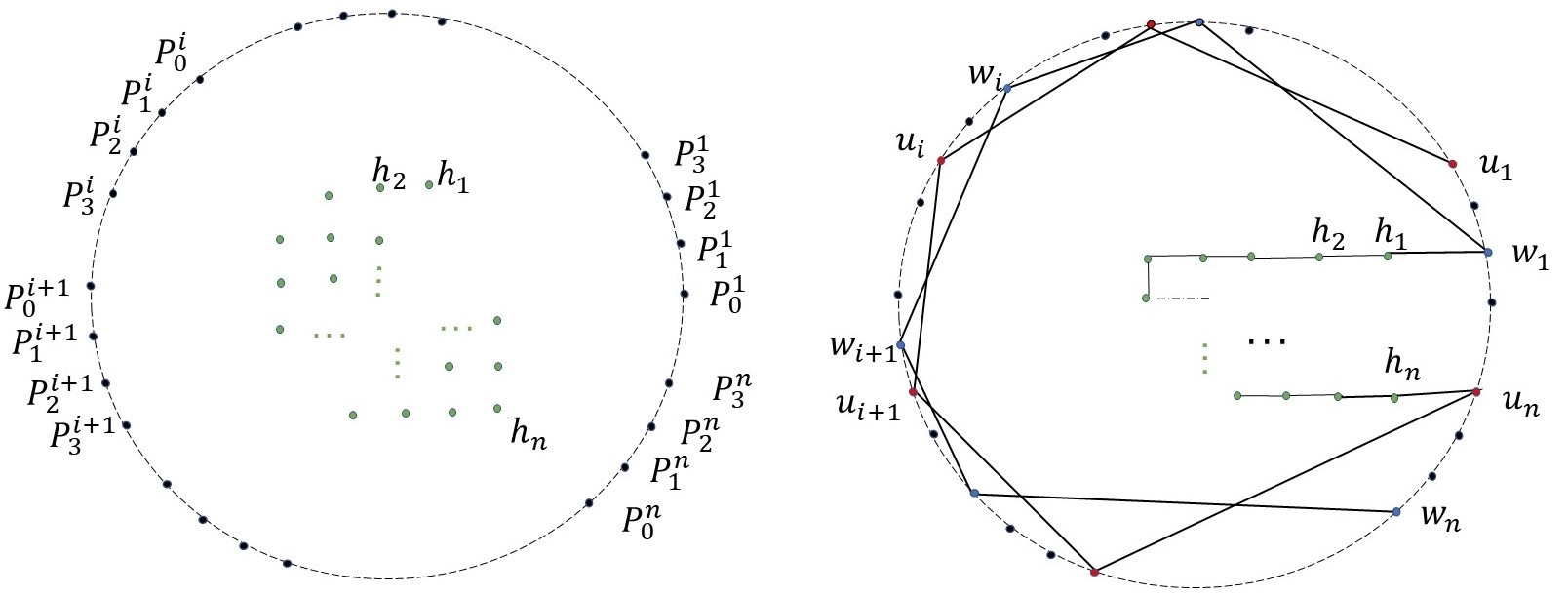}
\caption{Reduction from Set Disjointness to Convex Hull \label{Fig:convex}}
\end{figure}

We next describe the edges.
\begin{itemize}
\item
For each $1\leq i \leq n$, $E$ contains  $\{w_i, w_{i+1}\}$ and $\{u_i, u_{i+1}\}$.
\item
 For each $1\leq i \leq n$, $E$ contains $\{h_i, h_{i+1}\}$. 
 \item
$E$ contains $\{w_1, h_1\}$ and $\{h_{n}, u_n\}$.
\end{itemize}

\begin{claim}
\label{claim:ch-lb}
 $A\cap B\neq \emptyset $ if and only if there is some $i$ for which $w_i$ is on $P^i_1$ and the location of its neighbor on the convex hull of {\gab} is on position $P^i_3$ in the clockwise direction. Similarly, $A\cap B\neq \emptyset $ if and only if there is some $i$ for which $u_i$ is on $P^i_3$ and the location of its neighbor on the convex hull of {\gab} is on position $P^i_1$ in the counterclockwise direction.
\end{claim}
\begin{proof} 
In the solution of Convex Hull for these points, all the $u_i$'s and $w_i$'s are on the convex hull and no $h_i$ is on the convex hull. For any $1 \leq i \leq n$, $u_i$ is the neighbor of $w_i$ on the convex hull of {\gab} in the counterclockwise direction. So, in the solution of Convex Hull problem, if there is some $i$ for which $w_i$ is on $P^i_1$ and the location of its neighbor on the convex hull of {\gab} is on position $P^i_3$ in the clockwise direction, this means $u_i$ is on $P^i_3$. So, both $a_i=1$ and $b_i=1$.
Similarly we can show that $A\cap B\neq \emptyset $ if and only if there is some $i$ for which $u_i$ is on $P^i_3$ and the location of its neighbor on the convex hull of {\gab} is on position $P^i_1$ in the counterclockwise direction.\end{proof}

\begin{theorem} 
Any randomized asynchronous distributed algorithm for solving the Convex Hull in a graph of $n$ nodes requires $\Omega(n^2)$ expected bits of communication.
\end{theorem}
 \begin{proof}
Suppose  Alice and Bob are two endpoints of a path with $N=n/3$ edges, where Alice and Bob each know subsets $A$ and $B$ resp., of $\{1,...,N\}$, and they want to compute the Set Disjointness problem on $(a,b)$, where $a=(a_1,...,a_N)$ and $b=(b_1,...,b_N)$ are the characteristic vectors of $A$ and $B$.  
  Alice can construct the subgraph of $G$ induced by $H \cup W$ and given $B$, Bob can construct the subgraph induced by $H \cup U$, and they both know the candidate positions $P$.  They can then each carry out a simulation of the protocol which computes the convex hull of vertices of $G$ and output 1 or 0 depending on the answer of Convex Hull of {\gab} by using Claim \ref{claim:ch-lb}. Since Alice knows the states of all $w_i$ and each $w_i$ knows the location of its neighbors on the convex hull, Alice outputs 1 if and only if there is some $i$ for which for  $w_i$ is on $P^i_1$ and the location of its neighbor on the convex hull of {\gab} is on position $P^i_3$ in the clockwise direction. Similarly Bob outputs 1 if and only if there is some $i$ for which $u_i$ is on $P^i_3$ and the location of its neighbor on the convex hull of {\gab} is on position $P^i_1$ in the counterclockwise direction. Thus the communication cost is no less than the cost of computing disjoint sets over a path of length $N$ where $n=3N$ is the total number of nodes in $G$. By Theorem \ref{t:2-to-path},  $\Omega(n^2/\lg n)$ messages are required.\end{proof}  


\section{Closest Pair}
We present an approximation algorithm that gives us a tradeoff between the approximation ratio and the amount of communication. Then, we show that any distributed algorithm for solving Closest Pair in a graph of $n$ nodes on an $[n^c] \times [n^c]$ grid requires $\Omega (n^2)$ expected bits of communication to approximate within a $\frac{n^{c-1/2}}{4}$ factor of the optimum, for a constant $c>1+1/(2\lg n)$.  

\subsection{Approximate Closest Pair}
 
\begin{theorem}
There exists an asynchronous algorithm in CONGEST KT1 which computes an $\frac{n^c}{\sqrt{\frac{n-1}{2}}}$- Approximation Closest Pair in a graph of $n$ nodes on an $[n^c] \times [n^c]$ grid , for a constant $c>1/2$, with $O(n \lg n)$ messages and time $O(diam(\st) \lg n)$ where $diam(\st)$ is the diameter of the spanning tree, plus the costs for building the $\st$. \end{theorem} 

\begin{proof}

The main idea of algorithm \ref{Algclosestapprox} is to divide the $[n^c] \times [n^c]$ grid to a $\sqrt{n}$ by $\sqrt{n-1}$ grid, for a total of $\sqrt{n^2-n}$ grid cells: $C={c_1,c_2,..., c_{\sqrt{n^2-n}}}$. There are $n$ nodes and fewer than $n$ cells, so there is a cell with more than one node. The algorithm finds such a cell and nodes and broadcasts them. 

A spanning tree is first constructed. Next, we use binary search to find a cell with more than two points in $\lg({\sqrt{n^2-n}})$ rounds. Let $loc(p)$ denote the position of node $p$ and ${index(p)}$ denote the index of the cell in $C$ that contains $loc(p)$. Let $mid(s,e)=\lfloor \frac{s+e}{2}\rfloor$ for any $s$ and $e$.

The leader broadcasts $<start><s_i,e_i>$ in round $i$ for an interval $I_i=[s_i,e_i]$ of the indices, and the goal is to find the number of nodes which are in cells whose indices are in the first half of the range $I_i$. Starting with the leaves of $\st$, each node $p$ computes $x_p$, the number of nodes in the subtree rooted at $p$ which are in cells whose indices are in the first half of the interval $I_i$,$[s_i, mid(s_i,e_i)]$, and sends $<s_i,e_i,x_p>$ to its parent. When the leader receives messages for a specific $I_i$ from all of its children, it computes $x_L$ and picks $I_{i+1}$ to be the half of the interval that has more nodes than the number of cells, and broadcasts $I_{i+1}$. When the leader receives information about the last two cells, it picks the cell $ce$ that contains at least two points and broadcasts $<cell><index(ce)>$ to find two points in $ce$. When a node $p$ receives $<cell><index(ce)>$ and $p\in ce$, it sends $<decision><loc(p)$ to inform the leader about its location.

 As described before, there is always a cell with more than one point. The cells are $\frac{n^c}{\sqrt{n}}$ by $\frac{n^c}{\sqrt{n-1}}$. So we can easily observe that there is a pair of points with distance of at most the diameter of the cell, $\frac{\sqrt{2}n^c}{\sqrt{n-1}}$, but the actual closest pair could be at distance 1 from each other. Therefore, the approximation ratio of this algorithm is $\frac{n^c}{\sqrt{\frac{n-1}{2}}}$.

We observe that each node sends $O(\lg \sqrt{n^2-n})$ messages of size $O(\lg n)$ , so the total number of messages is $n \lg \sqrt{n^2-n}$ or $O(n \lg n)$ plus the cost of $ST$. \end{proof}

\begin{algorithm}
\caption{ $\frac{n^c}{\sqrt{\frac{n-1}{2}}}$-approximation Closest Pair,c>1/2 \label{Algclosestapprox}}

\begin{algorithmic}[1]
\Procedure{Closest Pair}{$P$ is a set of $n$ nodes}

\State Find spanning tree $\st$ with leader L.
\State $L$ broadcasts  $<start><1,{\sqrt{n^2-n}}>$ when $\st$ is complete.
\State When a leaf $p$ receives $<start><s,e>$,  if $index(p)\in [s,mid(s,e) ]$ then $p$ sends $<s,e,1>$ to its parent; otherwise sends $<s,e,0>$.
\State For specific $s$ and $e$, when a non-leaf node $p$, including  the leader, receives $<s,e,x_q>$ from all of its children $q$, $p$ sends $<s,e,\sum_{q}{x_q}+(1-x_p)$ to its parent, where $x_p=0$ if $index(p)\in [s,mid(s,e)]$ and $x_p=1$ otherwise.
 \State For specific $s$ and $e$, when $L$ receives $<e,s,x_q>$ from all of its children $q$ and computes $x_L$,
 
 	\textbf{If}	$e-s=1$, $L$ broadcasts $<cell> <index(ce)>$ to its children, where $ce$ is one of $c_e$ or $c_s$ that contains at least two points.
 	
	\textbf{Else If} $x_L> mid(s,e)-s+1$, $L$ broadcasts $<start>< s,mid(s,e)>$.
	
	\textbf{Else} $L$ broadcasts $<start><mid(s,e)+1,e>$.

	\State When a node $p\neq L$ receives $<cell><ce>$, if $index(p)=index(ce)$ and $p$ has not sent more than one $<decision>$ message to its parent before, then $p$ sends$<decision><loc(p)>$ to its parent.
	\State When a non-leaf node $p$ receives $<decision><loc(q)>$, from one of its children $q$ if $p$ has not send more than one $<decision>$ message to its parent before, then $p$ sends$<decision><loc(q)>$ to its parent.
	\State When a $L$ receives $<decision><loc(q)>$, from one of its children $q$ if $p$ has already received another $<decision>$ message or $index(L)=index(ce)$, then $L$ broadcasts the location of these two points in $ce$ through $ST$.

 \EndProcedure
\end{algorithmic}

 \end{algorithm}

\begin{coro}
There exists an asynchronous algorithm in CONGEST KT1 which computes an $\frac{n^c}{\sqrt{\frac{k}{2}}}$-Approximate Closest Pair in a graph of $n$ nodes on an $[n^c] \times [n^c]$ grid , for a constant $c>1/2$, in $O(n \lg k)$ messages in time $O(diam(\st) \lg k)$, where $diam(\st)$ is the diameter of the spanning tree, plus the cost of computing $\st$, for any $k\leq {n-1}$.
\end{coro} 
\begin{proof}
If in Algorithm \ref{Algclosestapprox} we divide the grid to $\sqrt{k}$ by $\sqrt{k}$ cells evenly, each size of $\frac{n^c}{\sqrt{k}}$, the diameter of each cell is $\frac{n^c}{\sqrt{\frac{k}{2
}}}$, and so is the approximation ratio of the algorithm. Each node $p$ sends $O(k)$ messages of size $\lg k$, and the algorithm uses $O(nk)$ messages.\end{proof}

\subsection{Lower Bound}
We show a reduction from Path Set Disjointness to Closest Pair. Similar ideas were used in the streaming model \cite{lowerbound}.

\noindent
Given an instance of  the Path Set Disjointness problem, in which there are two players Alice and Bob where Alice has subset $A$ with characteristic vector $a=(a_1,...,a_n)$ and Bob has subset $B$ with characteristic vector $b=(b_1,...,b_n)$, we consider the following graph \gab$(V,E)$ where $|V|=3n$. 
Let $V= W\cup U\cup H$, $W=\set{w_1,w_2,...,w_n}$, $U=\set{u_1,u_2,...,u_n}$, $H=\set{h_1,h_2,...,h_{n}}$.  
The locations of the vertices of {\gab} are as follows:
\begin{itemize}
\item For each $1\leq i \leq n$, if $a_i=0$, $w_i$ is on coordinate $(2i,0)$; otherwise  on coordinate $(2i,1)$.
\item For each $1\leq i \leq n$, if $b_i=0$, $u_i$ is on coordinate $(2i,3)$; otherwise on coordinate $(2i,2)$.
\item  For each $1\leq i \leq n/2$, $h_i$ is located on coordinate $(2(i+n), 0)$, and for each $n/2+1\leq i \leq n$, $h_i$ is on coordinate $(2(i+n),3)$ .
\end{itemize}

We next describe the edges.
\begin{itemize}
\item
For each $1\leq i \leq n$, $E$ contains  $\{w_i, w_{i+1}\}$ and $\{u_i, u_{i+1}\}$.
\item
 For each $1\leq i \leq n$, $E$ contains $\{h_i, h_{i+1}\}$. 
 \item
$E$ contains $\{w_{n}, h_1\}$ and $\{h_{n}, u_1\}$.
\end{itemize}

\begin{claim}
The distance of the two closest vertices in the problem of Closest Pair on {\gab} is 1 if and only if $A\cap B\neq \emptyset $.
\end{claim}
\begin{proof}
It is clear that the distance between any two points in $W$ or any two points in $U$ is at least 2. The distance between any pair of points in $H\cup W$ and any pair of points in $H\cup U$ is also at least 2.
The distance between the points of $W$ and the points of $U$ is 1 if and only if there is an $i$ such that $w_i$ is on $(2i,1)$ and $u_i$ is on $(2i,2)$. This occurs only when $a_i=1$ and $b_i=1$, so $A\cap B\neq \emptyset$.\end{proof}

\begin{figure}
\centering
\includegraphics[scale=0.24]{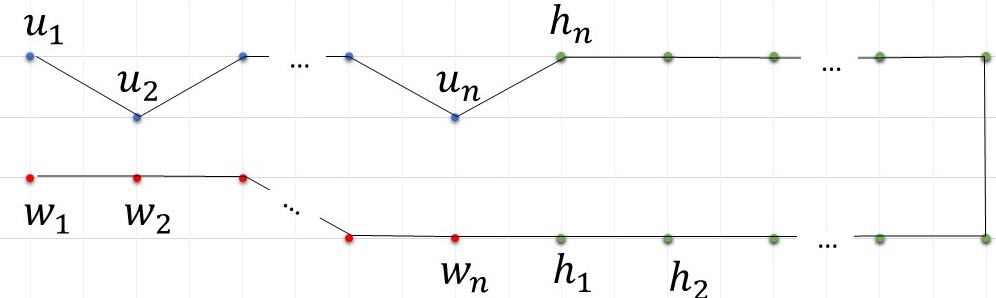}
\caption{Reduction from Set Disjointness to Closet Pair.  \label{Fig:cplb}}
\end{figure}
\begin{theorem}
Any randomized asynchronous distributed algorithm for solving Closest Pair in a graph with $n$ nodes on an $[n^c] \times [n^c]$ grid, for a constant $c>1+1/(2\lg n)$,  requires an expected $\Omega(n^2)$ bits of communication.
\end{theorem}
 \begin{proof}
  Suppose  Alice and Bob are two endpoints of a path with $n$ edges numbered $0,1,...,n$ where Alice and Bob each know subsets  $A$ and $B$ resp., of $\{1,...,n\}$, and they want to compute the Set Disjointness problem on $(a,b)$, where $a=(a_1,...,a_n)$ and $b=(b_1,...,b_n)$ are the characteristic vectors of $A$ and $B$.  
  Alice can construct the subgraph of $G$ induced by $H \cup W$ and given $B$, Bob can construct the subgraph induced by $H \cup U$.  They can then each carry out a simulation of the protocol which determines the closest pair of $G$ and output  0  iff  the closest pair is distance 1 away from each other.  Thus the communication cost is no less than the cost of computing disjoint sets over a path of length $N$ where $n=3N$ is the total number of nodes in $G$. By Theorem \ref{t:2-to-path},  $\Omega(n^2/\lg n)$ messages are required.\end{proof}
 
\begin{theorem}
Any randomized asynchronous distributed algorithm for approximating Closest in a graph with $n$ nodes on an $[n^c] \times [n^c]$ grid, for a constant $c>1+1/(2\lg n)$, within $\frac{n^{c-1}}{{2} }-\epsilon$ factor of optimal requires an expected $\Omega(n^2)$ bits of communications.
 \end{theorem}
 \begin{proof}
Consider the reduction shown in the previous section from Set Disjointness to Closest Pair with a slightly different location for the points: spreading points in a way that the minimum possible distance for a pair of points is 1 or $\frac{n^c}{{2}n}$ of each other in the same structure shown in Figure \ref{Fig:cplb}:
 if $k=\frac{n^{c-1}}{2}$
 \begin{itemize}
\item For each $1\leq i \leq n$, if $A_i=0$, $v_i$ is on coordinate $(ki,0)$; otherwise  on coordinate $(ki,k)$.
\item For each $1\leq i \leq n$, if $B_i=0$, $u_i$ is on coordinate $(ki,2k+1)$; otherwise on coordinate $(ki,k+1)$.
\item $h_i$ is located on coordinate $(k(i+n), 0)$ for $1\leq i \leq n/2$ and on  coordinate $(k(i+n), 2k+1)$ for $n/2+1\leq i \leq n$.
\end{itemize} 
If there as an approximation algorithm for Closest Pair with approximation ratio no more than  $\frac{n^c}{{2} n}-\epsilon$ and uses less $o(n^2)$ bits of communication, this means that the algorithm can distinguish between distance $\frac{n^c}{{2} n}-\epsilon$ and distance 1 and  answer Closest Pair exactly with using $o(n^2)$ bits of communication, which is a contradiction.\end{proof}

The following theorem shows a better hardness of approximation with the idea of spreading the points even more on the plane.
\begin{theorem}
Any randomized asynchronous distributed algorithm for approximating Closest Pair in a graph with $n$ nodes on an $[n^c] \times [n^c]$ grid, for a constant $c>1+1/(2\lg n)$, within  $\frac{n^{c-1/2}}{4}-\epsilon$ factor of optimal requires an expected $\Omega(n^2)$ bits of communications, where $c\geq1$ is a constant.

 \end{theorem} 
 \begin{proof}
When reducing Set Disjointness to Closest Pair, instead of locating all the $3n$ points on four horizontal lines, $y=0,1,2,3$, we locate the points on $\sqrt{n}$ groups of four horizontal lines each group containing $3\sqrt{n}$ points as shown in Figure \ref{Fig:CPimproved}. In this way the minimum possible distance between a pair of point is 1 or $\frac{n^c}{4 \sqrt{n}}$.
\end{proof} 

 \begin{figure}
\centering
\includegraphics[scale=0.25]{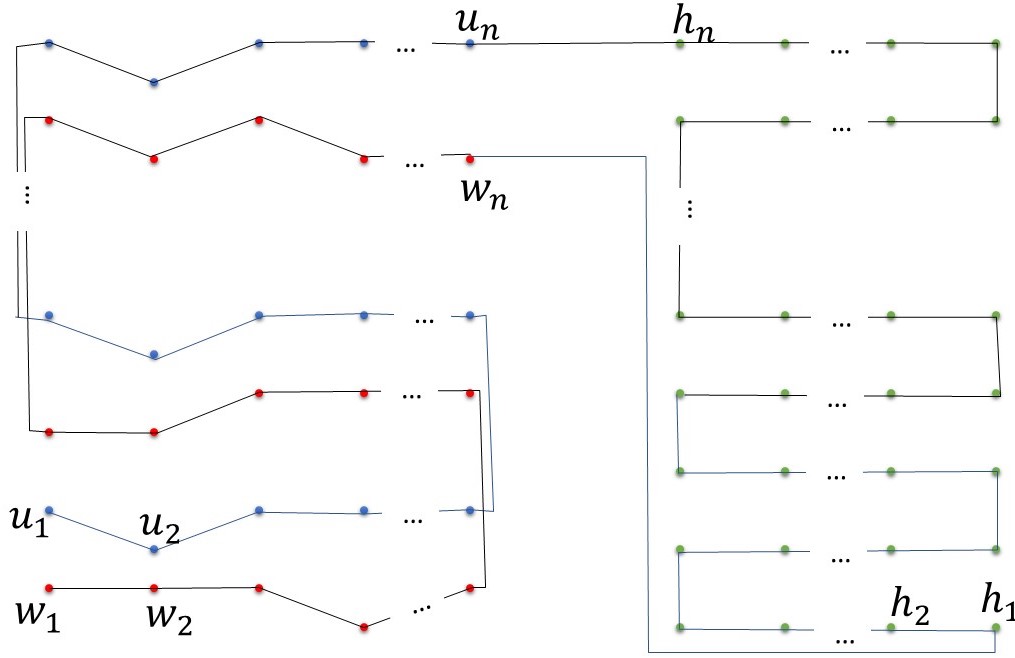}
\caption{Reduction from Set Disjointness to Closet Pair \label{Fig:CPimproved}}
\end{figure}

\section{open problems}

A bottleneck in our upper bounds is the computation of the spanning tree, which in the asynchronous model appears to require significantly more  communication (${\tilde{O}}(n^{3/2})$) than what is needed for the synchronous model ( $\tilde{O}(n)$ ). Is this gap necessary?

Finding an Euclidean minimum spanning tree (EMST) on a Geometric Network in the fixed point KT1 model, where edges of the spanning tree  are a subset of the edges of the geometric communication network and the weight of the edge between each pair of nodes is the Euclidean distance between the location of two nodes can be done with the algorithm found in \cite{KKT15}, since in this case each node knows the weights of the edges between itself and its neighbors.  But if a node's ID is not related to its position,  merely knowing the neighbors' IDs may be insufficient to solve this problem in $o(m)$ communication. 

\bibliographystyle{plain}
\bibliography{GeometricNetwork}

%
%

\end{document}